\documentclass[conference]{IEEEtran}
\IEEEoverridecommandlockouts
\usepackage{geometry}
\usepackage{setspace}
\usepackage{graphicx}
\usepackage{amssymb}
\usepackage{amsmath}
\usepackage{cite}
\usepackage{subfigure}
\usepackage{mathrsfs}
\usepackage[displaymath,mathlines]{lineno}
\usepackage{color}
\usepackage{tabulary}
\usepackage{multirow}
\usepackage{algpseudocode}
\usepackage{algorithm,algpseudocode}
\usepackage{algorithmicx}
\usepackage{pbox}
\usepackage{multicol}
\usepackage{lipsum}
\usepackage{amsthm}
\usepackage{relsize}
\usepackage{lipsum}
\usepackage{epstopdf}
\usepackage{mathtools}
\usepackage{calc}
\usepackage{makecell}
\usepackage{enumitem}

\makeatletter
\newcommand{\doublewidetilde}[1]{{%
		\mathpalette\double@widetilde{#1}}}
\newcommand{\double@widetilde}[2]{%
		\sbox\z@{$\m@th#1\widetilde{#2}$}%
		\ht\z@=.5\ht\z@
		\widetilde{\box\z@}}
\makeatother

\newtheorem{theorem}{Theorem}
\newtheorem{lemma}{Lemma}

\newtheorem{remark}{Remark}

\begin{document}
%
\title{\huge Power Allocation for Space-Terrestrial Cooperation Systems with Statistical CSI \vspace*{-0.4cm}}

\author{ \IEEEauthorblockN{ Trinh~Van~Chien$^{\nu}$, Eva~Lagunas$^{\ast}$, Tiep~M.~Hoang$^{\dagger}$,  Symeon~Chatzinotas$^{\ast}$,  Bj\"{o}rn~Ottersten$^{\ast}$, and Lajos~Hanzo$^{\xi}$} 
 \IEEEauthorblockA{$^{\nu}$School of Information and Communication Technology (SoICT), Hanoi University of Science and Technology, Vietnam  \\
 	$^{\ast}$Interdisciplinary Centre for Security, Reliability and Trust (SnT), University of Luxembourg, Luxembourg\\
 	$^{\dagger}$Department of Electrical Engineering,  University
 	of Colorado Denver, Denver, United States\\
 	$^{\xi}$School of Electronics and Computer Science, University of Southampton, Southampton, United Kingdom
 }
\vspace*{-1cm}
\thanks{This work has been supported by the Luxembourg National Research Fund (FNR) under the project MegaLEO (C20/IS/14767486).}
}

\maketitle

\begin{abstract}
This paper studies an integrated network design that boosts system capacity through cooperation between wireless access points (APs) and a satellite. By coherently combing the signals received by the central processing unit from the users through the space and terrestrial links, we mathematically derive an achievable throughput expression for the uplink (UL) data transmission over spatially correlated Rician channels. A closed-form expression is obtained when maximum ratio combining is employed to detect the desired signals. We formulate the max-min fairness and total transmit power optimization problems relying on the channel statistics to perform power allocation. The solution of each optimization problem is derived in form of a low-complexity iterative design, in which each data power variable is updated based on a  closed-form expression. The mathematical analysis is validated with numerical results showing the added benefits of considering a satellite link  in terms of improving the ergodic data throughput. 
\vspace*{-0.1cm}
\end{abstract}

%
\IEEEpeerreviewmaketitle

\section{Introduction}

Future wireless systems will offer high throughput per user principally based on the access to new spectrum, while intelligently coordinating a number of access points (APs) in a coverage area. This leads to the concept of Cell-Free Massive MIMO systems \cite{ngo2017cell}, which serve a group of users by a group of APs. 
The network then coherently combines different observations of the transmitted waves received over multiple heterogeneous propagation paths \cite{9500472}  using either maximum ratio combing (MRC) or minimum mean square error (MMSE) reception. Satellite communication has attracted renewed interest as a promising technique to provide services for many users across a large coverage area \cite{schwarz2019mimo}. The low latency, small size, and short delays of non-geostationary (NGSO) satellites has demonstrated their potentialities.  

Both academia and industry have recently intensified their research of NGSO aided terrestrial communications \cite{3gpp2019study}. As for the demands of tomorrow's networks, the authors of \cite{abdelsadek2021future} considered the performance of a space communication system that replaces terrestrial APs by LEO satellites. Despite integrating a LEO satellite into a terrestrial network \cite{riera2021enhancing}, the received signals were detected independently i.e., without exploiting the benefits of constructive received signal combination. As a further contribution, the coexistence of fixed satellite services and cellular networks was studied in \cite {du2018secure} for transmission over slow fading channels subject to individual user throughput constraints. The ergodic rate of the fast fading channels was considered in \cite{ruan2019energy} under the assumption of perfect channel state information (CSI) and no spatial correlation. In a nutshell, the literature of  space-terrestrial integrated networks suffers from the two limitations: $i)$ most of the performance analysis and resource allocation studies rely on the idealized simplifying assumptions of perfect instantaneous CSI knowledge, which is challenging to acquire in practice, especially under high mobility scenarios; and $ii)$ the spatial correlation between satellite antennas is ignored, despite its coherent existence in the planar antenna arrays.

\textit{Paper contributions}: By taking advantage of both the distributed Cell-Free Massive MIMO structure and satellite communications, we evaluate the ergodic throughput of each user relying on a limited number of APs and demonstrate how the satellite enhances the system performance. Explicitly, our main contributions are summarized as follows:
 $i)$ We derive the achievable rate expression of each user in the uplink (UL) for transmission over spatially correlated fading channels, when relying on centralized data processing. If the MRC technique is used  both at the APs and at the satellite gateway, a closed-form expression of the ergodic net throughput will be derived. $ii)$ Furthermore, we formulate a max-min fairness optimization problem that simultaneously allocates the  powers to all the scheduled users and guarantees uniform throughput for the entire network. We determine the user-specific optimal power for each user at a low complexity by exploiting the quasi-concavity of the objective function, the standard interference functions, and the bisection method. $iii)$ Our numerical results quantify the value of the satellite in improving both the total and the minimum user throughput.

\textit{Notation}: Upper and lower bold letters denote matrices and vectors, respectively. The superscript $(\cdot)^T$ and $(\cdot)^H$ are the regular and Hermitian transpose. $\mathrm{tr}(\mathbf{X})$ is the trace of square matrix $\mathbf{X}$, whilst $\mathbf{I}_N$ is an identity matrix of size $N \times N$. The circularly symmetric Gaussian distribution is denoted by $\mathcal{CN}(\cdot, \cdot)$, and the expectation of a random variable is $\mathbb{E}\{\cdot\}$. Moreover, $\mod(\cdot, \cdot)$ is the modulus operation, $\lfloor \cdot \rfloor$ is the floor function, and $\otimes$ is the Kronecker product. The cardinality of set $\mathcal{X}$ is denoted by $|\mathcal{X}|$. Finally, $J_1(\cdot)$ is the Bessel function of the first kind of order one.
\vspace{-0.2cm}
\section{System Model} \label{Sec:SysModel}
\vspace{-0.1cm}
We consider a  distributed multi-user network comprising $M$ APs each equipped with a single receiver antenna (RA). The APs  cooperatively serve $K$ users in the UL, all equipped with a single transmit antenna (TA). The system performance is enhanced by the assistance of an NGSO satellite having $N$ RAs arranged in an $N_H \times N_V$-element rectangular array ($N =N_H \times N_V$), as illustrated in Fig.~\ref{FigSysModel}. Both the satellite gateway and the APs forward the UL signals received from the users to a central processing unit (CPU) by fronthaul links. The APs rely on optical fronthaul links, while the satellite has a radio downlink (feeder link) to the ground station, which forwards the users' UL signal to the CPU. Since the dispersive channels fluctuate both time and frequency over wideband systems, orthogonal frequency division multiplexing (OFDM) is used for mitigating it \cite{you2020massive}. A block-fading channel model is applied across the OFDM symbols, where the fading envelope is assumed to be frequency-flat through an entire OFDM symbol and then faded randomly for the next OFDM symbol. We assume that a fraction of $K$ subcarriers of each  OFDM symbol are known pilots, while the remaining $\tau_c - K$ subcarriers are used for UL payload data transmission. The satellite antenna's gain is sufficiently high to amplify the weak UL signals received from the distant terrestrial users \cite{3gpp2019study}. 
\begin{figure}[t]
	\centering
	\includegraphics[trim=2.6cm 2.6cm 2.2cm 1.7cm, clip=true, width=2.8in]{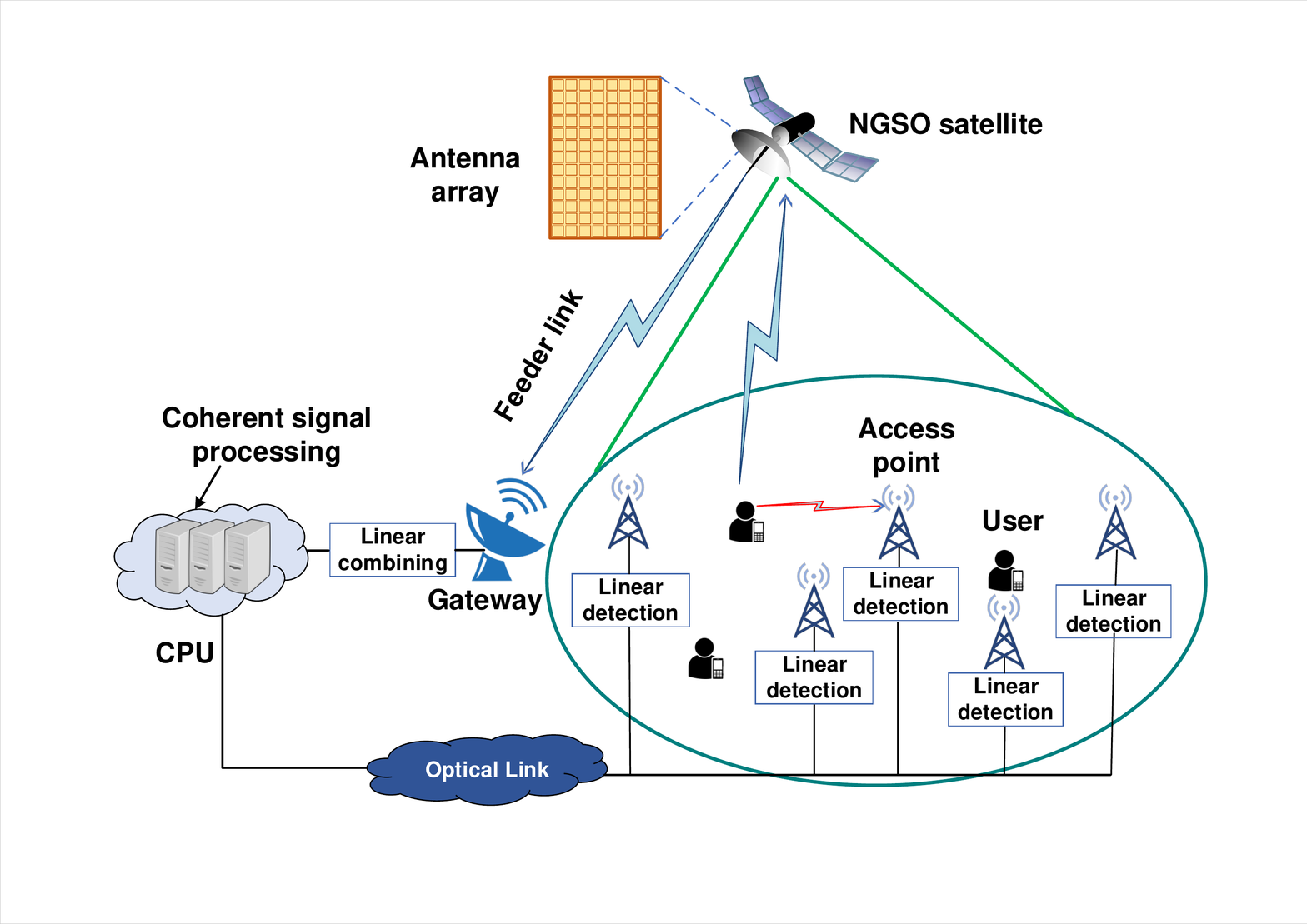} \vspace*{-0.2cm}
	\caption{Illustration of a cooperative satellite-terrestrial  wireless network.}
	\label{FigSysModel}
	\vspace*{-0.5cm}
\end{figure}
\vspace{-0.2cm}
\subsection{Channel Model} \label{Sec:ChannelModel}
\vspace{-0.1cm}
The terrestrial UL channel between AP~$m$ and user~$k$, $\forall m,k,$ denoted by $g_{mk} \in \mathbb{C}$ is modeled as
$g_{mk}  \sim \mathcal{CN}(0, \beta_{mk})$,
where $\beta_{mk}$ is the large-scale fading coefficient. The channel between the UL transmitter of user~$k$ and the satellite receiver, denoted by $\mathbf{g}_k \in \mathbb{C}^N$, has been modeled according to the 3GPP recommendation (Release 15) \cite{3gpp2019study} and obeys the Rician distribution as
$\mathbf{g}_k \sim \mathcal{CN}(\bar{\mathbf{g}}_k, \mathbf{R}_k)$,
where $\bar{\mathbf{g}}_k \in \mathbb{C}^N$ denotes the LoS components gleaned from the $N$ RAs in the UL. The matrix $\mathbf{R}_k \in \mathbb{C}^{N \times N}$ is the covariance matrix of the spatially correlated signals collected by the RAs of the satellite attenuated by the propagation loss.\footnote{The propagation loss in the carrier frequency range from $0.5$~GHz to $100$~GHz has been well documented in \cite{3gpp2019study}.} The LoS component is given by 
\vspace*{-0.2cm}
\begin{equation} \label{eq:bargk}
\bar{\mathbf{g}}_k = \sqrt{\kappa_k \beta_k/(\kappa_k+1)} \big[e^{j \pmb{\ell}(\theta_k, \omega_k)^T \mathbf{c}_1}, \ldots,  e^{j \pmb{\ell}(\theta_k, \omega_k)^T \mathbf{c}_N}  \big]^T,
\vspace*{-0.2cm}
\end{equation}
where $\theta_k$ and $\omega_k$ are the elevation and azimuth angle, respectively; $\kappa_k \geq 0$ represents the Rician factor;  and $\beta_k$ is the large-scale fading coefficient encountered between user~$k$ and the satellite. The antenna array is fabricated in a rectangular surface whose wave form vector $\pmb{\ell}(\theta_k, \omega_k)$ \cite{massivemimobook,Chien2021TWC} is
\vspace*{-0.2cm}
\begin{equation}
\pmb{\ell}(\theta_k, \omega_k) = \frac{2 \pi}{\lambda}[ \cos(\theta_k)\cos(\omega_k), \sin(\theta_k)\cos(\omega_k), \sin(\theta_k) ]^T,
\vspace*{-0.2cm}
\end{equation}
with $\lambda$ being the wavelength of the carrier. In \eqref{eq:bargk}, there are $N$ indexing vectors, each given by
$\mathbf{c}_m = [0, \mathrm{mod}(n-1, N_H) d_H, \lfloor (n-1)/N_H \rfloor d_V ]^T$,
where $d_H$ and $d_V$ represent the antenna spacing in the horizontal and vertical direction, respectively. The 3D channel model of \cite{ying2014kronecker,chatzinotas2009multicell} relies on the spatial correlation matrices  of a planar antenna array, formulated as
$\mathbf{R}_k =  (\beta_k/(\kappa_k +1)) \mathbf{R}_{k,H} \otimes \mathbf{R}_{k,V}$, 
where $\mathbf{R}_{k,H} \in \mathbb{C}^{N_H \times N_H}$ and $\mathbf{R}_{k,V} \in \mathbb{C}^{N_V \times N_V}$ are the spatial correlation matrices along the horizontal and vertical direction.

\subsection{Uplink Pilot Training} \label{Sec:ULP}
\vspace{-0.2cm}
 All the $K$ users simultaneously transmit their pilot signals in each coherence block of the UL. We assume to have the same number of orthonormal pilots as users, i.e., we have the set $\{ \pmb{\phi}_1, \ldots, \pmb{\phi}_K\}$, where the pilot $\pmb{\phi}_k \in \mathbb{C}^K$ is assigned to user~$k$. 
The training signal received at AP~$m$, $\mathbf{y}_{pm} \in \mathbb{C}^K$, is a superposition of all the UL pilot signals transmitted over the propagation environment, which is formulated as
\vspace*{-0.2cm}
\begin{equation} \label{eq:ypm}
	\mathbf{y}_{pm}^H = \sum\nolimits_{k=1}^K \sqrt{pK} g_{mk} \pmb{\phi}_k^H + \mathbf{w}_{pm}^H,
\vspace*{-0.2cm}
\end{equation}
where $p$ is the transmit power allocated to each pilot symbol and $\mathbf{w}_{pm} \sim \mathcal{CN}(\mathbf{0}, \sigma_a^2 \mathbf{I}_{\tau_p})$ is the additive white Gaussian noise (AWGN) at AP~$m$ having zero mean and standard derivation of $\sigma_a$~[dB]. Furthermore, the training signal received at the satellite gateway of Fig.~\ref{FigSysModel} used for estimating the satellite UL channel is formulated as
\vspace*{-0.2cm}
\begin{equation} \label{eq:Yp}
\mathbf{Y}_p = \sum\nolimits_{k=1}^K \sqrt{pK} \mathbf{g}_{k} \pmb{\phi}_k^H + \mathbf{W}_{p},
\vspace*{-0.2cm}
\end{equation}
where $\mathbf{W}_{p} \in \mathbb{C}^{N \times K}$ is the AWGN with each element distributed as $\mathcal{CN}(0, \sigma_s^2)$. The desired UL channels are estimated at the APs and the satellite gateway by the MMSE estimation.
\vspace*{-0.2cm}
\begin{lemma} \label{Lemma:Est}
The MMSE estimate $\hat{g}_{mk}$ of the UL channel $g_{mk}$ between user~$k$ and AP~$m$ can be computed from \eqref{eq:ypm} as
\vspace*{-0.2cm}
\begin{equation} \label{eq:ChanEstgmk}
\hat{g}_{mk} = \mathbb{E}\{ g_{mk} | \mathbf{y}_{pm}^H\pmb{\phi}_k \} = \sqrt{pK} \beta_{mk}\mathbf{y}_{pm}^H\pmb{\phi}_k /(p K  \beta_{mk} + \sigma_a^2),
\vspace*{-0.2cm}
\end{equation}
which is distributed as $\hat{g}_{mk} \sim \mathcal{CN}(0, \gamma_{mk})$ and its variance is $\gamma_{mk} = \mathbb{E} \{ |\hat{g}_{mk}|^2\} =   pK \beta_{mk}^2 / ( p K  \beta_{mk} + \sigma_s^2)$. The channel estimation error $e_{mk} = g_{mk} - \hat{g}_{mk}$ is distributed as $e_{mk} \sim \mathcal{CN}(0, \beta_{mk} - \gamma_{mk} )$. The MMSE estimate $\hat{\mathbf{g}}_k$ of the channel $\mathbf{g}_k$  can be formulated based on \eqref{eq:Yp} as
\vspace*{-0.2cm}
\begin{equation} \label{eq:ChannelEstgk}
\hat{\mathbf{g}}_k = \bar{\mathbf{g}}_k +  \sqrt{pK} \mathbf{R}_k \pmb{\Phi}_k \big( \mathbf{Y}_{p} \pmb{\phi}_k -   \sqrt{pK} \bar{\mathbf{g}}_k \big),
\vspace*{-0.2cm}
\end{equation}
where we have $\pmb{\Phi}_k = \big( p K \mathbf{R}_{k} + \sigma_s^2 \mathbf{I}_N \big)^{-1}$. Additionally, the channel estimation error $\mathbf{e}_k = \mathbf{g}_k - \hat{\mathbf{g}}_k$ and the channel estimate $\hat{\mathbf{g}}_k$ are independent random variables, which are distributed as $\hat{\mathbf{g}}_k \sim \mathcal{CN}(\bar{\mathbf{g}}_k,  p K \mathbf{R}_k \pmb{\Phi}_k \mathbf{R}_k), \mathbf{e}_k \sim \mathcal{CN}(\mathbf{0}, \mathbf{R}_k - p K \pmb{\Theta}_k),$ with $\pmb{\Theta}_k = \mathbf{R}_k \pmb{\Phi}_k \mathbf{R}_k, \forall k$.
\vspace*{-0.2cm}
\end{lemma}
\begin{proof}
The proof follows from adopting the standard MMSE estimation of \cite{massivemimobook} for our system and channel model.
\vspace*{-0.2cm}
\end{proof}
Given the independence of the channel estimates and estimation errors, this may be  conveniently exploited in our ergodic data throughput analysis and optimization in the next sections. The LoS components of the satellite links can be estimated very accurately at the satellite gateway from its received training signals.
\vspace{-0.1cm} 
\section{Uplink Data Transmission and Ergodic Throughput Analysis} \label{Sec:ULData}
\vspace{-0.1cm}
This section presents the UL data transmission, where all users send  their signals both to the APs and to the satellite. The UL throughput of each user is derived in a closed-form expression for an MRC receiver.
\vspace{-0.1cm}
\subsection{Uplink Data Transmission} \label{Sec:UDT}
\vspace{-0.1cm}
All the $K$ users transmit their data both to the $M$ APs and to the satellite, where the symbol $s_k$  of user~$k$ obeys  $\mathbb{E}\{ |s_k|^2\} =1$ and has a transmit power level $\rho_k >0$. The  signal received at the satellite, denoted by $\mathbf{y} \in \mathbb{C}^N$, and AP~$m$, denoted by $y_m \in \mathbb{C}$, are
\vspace*{-0.2cm}
\begin{equation} \label{eq:ReceiveSigAP}
\begin{cases}
\mathbf{y} &= \sum\nolimits_{k=1}^K \sqrt{\rho_k} \mathbf{g}_k s_k  + \mathbf{w}, \\
y_m &=  \sum\nolimits_{k=1}^K \sqrt{\rho_k} g_{mk} s_k  + w_m, 
 \end{cases}
\vspace*{-0.2cm}
\end{equation}
where $\mathbf{w} \sim \mathcal{CN}(\mathbf{0}, \sigma_s^2 \mathbf{I}_N)$ and $w_m \sim \mathcal{CN}(0,\sigma_a^2)$  represent the AWGN noise at the satellite receiver and at AP~$m$, respectively. By exploiting \eqref{eq:ReceiveSigAP}, the system detects the  signal received from user~$k$ at the CPU  from the following expression
\vspace*{-0.2cm}
\begin{equation} \label{eq:2ndsk}
\hat{s}_k  = \mathbf{u}_k^H \mathbf{y} + u_{mk} y_m,
\vspace*{-0.2cm}
\end{equation}
where $\mathbf{u}_k \in \mathbb{C}^{N}$ is the linear detection vector used for inferring the desired signal arriving from the satellite and $u_{mk} \in \mathbb{C}$ is the detection coefficient used by AP~$m$. The received symbol estimate in \eqref{eq:2ndsk} combines all the different propagation paths, which explicitly unveils the potential benefits of integrating a satellite into terrestrial networks for improving the reliability and/or the throughput. Upon considering only one of the right-hand side terms of \eqref{eq:2ndsk}, the received signal becomes that of a conventional satellite network \cite{abdelsadek2021future} or a terrestrial cooperative network \cite{ngo2017cell}. Thus, we are considering an  advanced cooperative wireless network relying on the coexistence of both space and terrestrial links.
\vspace{-0.3cm}
\subsection{Uplink Throughput}
\vspace{-0.2cm}
We emphasize that if the number of APs and antennas at the satellite is sufficiently high to treat the channel gain of the desired signal in \eqref{eq:2ndsk} as a deterministic value, the  throughput of user~$k$ can be analyzed conveniently. In order to carry out the throughput analysis, let us first introduce the new variable
\vspace*{-0.2cm}
\begin{equation} \label{eq:zkkprime}
z_{kk'} = \mathbf{u}_k^H \mathbf{g}_{k'}  +  \sum\nolimits_{m=1}^M  u_{mk}^\ast  g_{mk'},
\vspace*{-0.2cm}
\end{equation}
which we term as the overall channel coefficient after the use of signal detection techniques, including both the satellite and terrestrial effects. The overall channel coefficient in \eqref{eq:zkkprime} leads to a coherent received signal combination at the CPU, where we assume perfectly phase-coherent symbol-synchronization at all the receivers and neglect any phase-jitter.\footnote{The longer propagation delay of the concatenated user-satellite-ground-station path has to be compensated by appropriately delaying the terrestrial signal for coherent combination at the CPU \cite{you2020massive}. For a LEO satellite at 600~km altitude this only imposes $4$~ms two-hop delay.} In particular, the desired signal in \eqref{eq:2ndsk} becomes
\vspace*{-0.2cm}
\begin{equation} \label{eq:Decodesig}
	\begin{split}
		& \hat{s}_k  = \sqrt{\rho_k} \mathbb{E} \{ z_{kk} \} s_k +  \sqrt{\rho_k} \left( z_{kk}  -  \mathbb{E} \{ z_{kk} \}  \right)s_k +  \\
		& \sum\nolimits_{k'=1, k'\neq k}^K \sqrt{\rho_{k'}} z_{kk'} s_{k'} + \alpha_k^\ast \mathbf{u}_k^H \mathbf{w} + \sum\nolimits_{m=1}^M u_{mk}^\ast w_m,
	\end{split}
\vspace*{-0.2cm}
\end{equation}
where the first additive term contains the desired signal associated with a deterministic effective channel gain. The second term represents the  beamforming uncertainty, demonstrating the randomness of the effective channel gain for a given signal detection technique. The remaining terms are the mutual interference and noise. By virtue of the  use-and-then-forget channel capacity bounding technique of \cite{massivemimobook}, the  ergodic  throughput of user~$k$ is
\vspace*{-0.2cm}
\begin{equation} \label{eq:Rkv1}
	R_k = \left( 1 - K/\tau_c \right) B \log_2 ( 1 + \mathrm{SINR}_k ), \mbox{[Mbps]},
\vspace*{-0.2cm}
\end{equation}
where $B$ [MHz] is the system bandwidth. The effective signal-to-interference-and-noise ratio (SINR) expression, denoted by $\mathrm{SINR}_k$, is given as in \eqref{eq:SINRk}.
\begin{figure*}
\vspace*{-0.2cm}
\begin{equation} \label{eq:SINRk}
\mathrm{SINR}_k = \frac{\rho_k | \mathbb{E}\{ z_{kk}\}|^2 }{\sum\nolimits_{k'=1}^K \rho_{k'}  \mathbb{E}\{ |z_{kk'}|^2 \} - \rho_k \big| \mathbb{E}\{ z_{kk}\} \big|^2 +  \mathbb{E} \big\{ \big| \mathbf{u}_k^H \mathbf{w} \big|^2 \big\} + \sum\nolimits_{m=1}^M  \mathbb{E} \big\{ | u_{mk}^\ast w_m |^2 \big\}  }
\vspace*{-0.1cm}
\end{equation}
\hrule
\vspace*{-0.4cm}
\end{figure*}
We stress that the throughput in \eqref{eq:Rkv1} can be achieved by arbitrary signal detection techniques at the satellite and APs, since it is a lower bound on the channel capacity. One can evaluate \eqref{eq:Rkv1} numerically, but requires many realizations of the small-scale fading coefficients to compute several expectations. The direct evaluation of \eqref{eq:Rkv1} by Monte Carlo simulations does not provide analytical insights about the impact of the individual parameters on the system performance. 
\vspace{-0.2cm}
\subsection{Uplink Throughput for Maximum Ratio Combining}
\vspace{-0.1cm}
For gaining further insights, we derive a  closed-form expression for \eqref{eq:Rkv1} by relying on statistical signal processing, when the MRC receiver is used by both the satellite and the AP, i.e., $u_{mk} = \hat{g}_{mk}, \forall m,k,$ and $\mathbf{u}_k = \hat{\mathbf{g}}_k, \forall k$. 
\begin{theorem} \label{Theorem:ClosedForm}
If the MRC receiver is utilized for detecting the desired signal, the UL throughput of user~$k$ is evaluated by  \eqref{eq:Rkv1} with the effective SINR value obtained in closed form as
\vspace*{-0.2cm}
\begin{equation} \label{eq:ClosedSINR}
 \mathrm{SINR}_k = \frac{\rho_k \big(\|\bar{\mathbf{g}}_k\|^2 +  p K \mathrm{tr}(\pmb{\Theta}_k)  +  \sum_{m=1}^M \gamma_{mk} \big)^2}{ \mathsf{MI}_k + \mathsf{NO}_k},
\vspace*{-0.2cm}
\end{equation}
where the mutual interference $\mathsf{MI}_k$, and noise $\mathsf{NO}_k$ are respectively given as follows
\vspace*{-0.2cm}
\begin{align}
\mathsf{MI}_k = & \sum\nolimits_{k' =1 , k' \neq k }^K \rho_{k'} |\bar{\mathbf{g}}_{k}^H  \bar{\mathbf{g}}_{k'} |^2  + p K \sum\nolimits_{k' =1}^K \rho_{k'}    \bar{\mathbf{g}}_{k'}^H\pmb{\Theta}_k \bar{\mathbf{g}}_{k'}  \notag \\
& + \sum\nolimits_{k' =1}^K \rho_{k'} \bar{\mathbf{g}}_{k}^H \mathbf{R}_{k'} \bar{\mathbf{g}}_{k}  +  p K \sum\nolimits_{k' =1 }^K \rho_{k'}   \mathrm{tr}( \mathbf{R}_{k'} \pmb{\Theta}_k )   \notag\\
& + \sum\nolimits_{k' =1}^K \sum\nolimits_{m=1}^M \rho_{k'}  \gamma_{mk} \beta_{mk'}, \label{eq:MIk}\\
\mathsf{NO}_k =&  \sigma_s^2 \|\bar{\mathbf{g}}_k\|^2 +  p K \sigma_s^2 \mathrm{tr}( \pmb{\Theta}_k ) + \sigma_a^2 \sum\nolimits_{m=1}^M  \gamma_{mk}. \label{eq:NOk}
\vspace*{-0.2cm}
\end{align}
\end{theorem}
\begin{proof}
The proof computes the expectations in \eqref{eq:SINRk} using the channel models in Section~\ref{Sec:ChannelModel} and the statistical information in Lemma~\ref{Lemma:Est} and it is omitted due to space limitations.
\vspace*{-0.2cm}
\end{proof}
The UL throughput of user~$k$ obtained in Theorem~\ref{Theorem:ClosedForm} is a function of the channel statistics, which has a complex expression due to the presence of space links. The spatial correlation and the LoS components created by the presence of the satellite beneficially boost the desired signals, as shown in the numerator of \eqref{eq:ClosedSINR}. The denominator of \eqref{eq:ClosedSINR} represents the interference and noise that degrades the performance, where the SINR is linearly proportional both to  the number of satellite antennas and  APs. 
\begin{remark} 
 The coexistence of the satellite and APs together with the coherent data processing at the CPU yield a quadratic array gain on the order of $(M + 2N)^2$. The closed-form expression of the ergodic data throughput in \eqref{eq:ClosedSINR} quantifies the  improvements offered by space-terrestrial communications. The stand-alone terrestrial communications only provides an array gain scaling increased with the number of APs, i.e., say $M^2$, while the dominant LoS path in each satellite link  boosts the array gains with the order of $4N^2$. 
\vspace*{-0.2cm}
\end{remark} 
\vspace*{-0.2cm}
\section{Uplink Data Power Allocation for Space-Terrestrial Communications} \label{Sec:PowerAllocation}
\vspace*{-0.2cm}
This section considers the max-min fairness optimization problem, which underlines the considerable benefits of a collaboration between the space and terrestrial links under a finite transmit power at each user.
\vspace*{-0.2cm}
\subsection{Max-Min Fairness Optimization}
\vspace*{-0.2cm}
Fairness is of paramount importance for planning the networks to provide an adequate throughput for all users by maximizing the lowest achievable ergodic rate. The max-min fairness optimization is formulated as
\vspace*{-0.2cm}
\begin{subequations} \label{Problem:MaxMinQoS}
	\begin{alignat}{2}
		& \underset{ \{ \rho_{k} \} }{\textrm{maximize}} \; \underset{k}{\textrm{min}}
		& & \, \, R_k \label{eq:Obj1} \\
		& \textrm{subject to}
		& & 0 \leq \rho_{k} \leq P_{\mathrm{max},k} \;, \forall k  ,
	\end{alignat}
\vspace*{-0.1cm}
\end{subequations}
where $P_{\mathrm{max},k}$ is the maximum power that user~$k$ can allocate to each data symbol. Due to the universality of the data throughput expression of \eqref{eq:Rkv1}, Problem~\eqref{Problem:MaxMinQoS} is applicable to any linear receiver combining method. This paper focuses on the MRC method because of the closed-form SINR expression for each user  shown in \eqref{eq:ClosedSINR}.\footnote{An extension to the other linear combining technique can be accomplished by using the same methodology, but may require extra cost to evaluate the expectations in \eqref{eq:SINRk} numerically.} The main features of \eqref{Problem:MaxMinQoS} are given in Lemma~\ref{Lemma:QuasiConcave}. 
\vspace*{-0.2cm}
\begin{lemma} \label{Lemma:QuasiConcave}
 Problem~\eqref{Problem:MaxMinQoS} is  quasi-concave as the objective function is constructed based on the ergodic UL throughput in \eqref{eq:Rkv1} with the SINR expression in \eqref{eq:ClosedSINR}.
 \vspace*{-0.2cm}
\end{lemma}
\begin{proof}
The proof is based on the definition of the upper-level set for a quasi-concave problem. The detailed proof is omitted due to space limitations.
\vspace*{-0.1cm}
\end{proof}
From the results obtained by Lemma~\ref{Lemma:QuasiConcave},  we exploit the quasi-concavity to find the most energy-efficient solution. Upon exploiting that $\mathrm{SINR}_k = 2^{\tau_c R_k /(B(\tau_c -K))}, \forall k$,  Problem~\eqref{Problem:MaxMinQoS} is reformulated in an equivalent form by exploiting the epigraph representation of \cite[page 134]{Boyd2004a} as follows
\vspace*{-0.2cm}
\begin{equation} \label{Problem:Epigraphform}
	\begin{aligned}
		& \underset{ \{ \rho_{k} \} }{\textrm{maximize}} 
		& &  \xi \\
		& \textrm{subject to} && \mathrm{SINR}_k \geq \xi , \forall k, \\
		& & & 0 \leq \rho_{k} \leq P_{\mathrm{max},k} \;, \forall k.
	\end{aligned}
\vspace*{-0.2cm}
\end{equation}
Problem~\eqref{Problem:Epigraphform} could be viewed as a geometric program to attain the maximal fairness level, but this would impose high computational complexity, since a hidden convex structure should be deployed \cite{van2018joint}. Observe that, for a given value of $\xi = \xi_o$ in the feasible domain, the minimum total transmit power consumption is obtained by the solution of the following optimization problem
\vspace*{-0.2cm}
\begin{subequations}\label{Problem:TotalTransmitPower}
	\begin{alignat}{2}
		& \underset{ \{ \rho_{k} \} }{\textrm{minimize}} 
		& &  \sum\nolimits_{k=1}^K \rho_k \\
		& \textrm{subject to} && \, \, \mathrm{SINR}_k \geq \xi_o , \forall k, \label{eq:SINRConstraints}\\
		& & & 0 \leq \rho_{k} \leq P_{\mathrm{max},k} \;, \forall k.
	\end{alignat}
\end{subequations}
After that, the most energy-efficient solution of Problem~\eqref{Problem:Epigraphform} should be obtained by finding the maximum value of the variable $\xi$ of using, for example, the popular bisection method.  \eqref{Problem:TotalTransmitPower} is  a linear program, so a canonical algorithm can get the global solution by the classic interior-point method. The main cost in each iteration is associated with computing the first derivative of the SINR constraints \eqref{eq:SINRConstraints}, which might still impose high computational complexity. Subsequently, in this paper, we propose a low complexity algorithm based on the alternating optimization approach and the closed-form solution for each power coefficient by virtue of the standard interference function \cite{van2021uplink}. By stacking all the transmit data powers in a vector $\pmb{\rho} = [\rho_1, \ldots, \rho_K] \in \mathbb{R}_+^K$, the SINR constraint of user~$k$ is reformulated as
$\rho_k \geq I_k (\pmb{\rho})$,
where  the standard interference function $I_k (\pmb{\rho})$ is defined as
\vspace*{-0.2cm}
\begin{equation}\label{eq:Ikrho}
I_k (\pmb{\rho}) =\frac{\xi_o\mathsf{MI}_k (\pmb{\rho}) + \xi_o \mathsf{NO}_k}{\left|\|\bar{\mathbf{g}}_k\|^2 +  p K \mathrm{tr}(\pmb{\Theta}_k)  +  \sum\nolimits_{m=1}^M \gamma_{mk} \right|^2},
\vspace*{-0.2cm}
\end{equation}
where the detailed expression of $\mathsf{MI}_k (\pmb{\rho})$ has been given in \eqref{eq:MIk}, but here we express it as a function of the transmit power variables stacked in $\pmb{\rho}$. Apart from the SINR constraint, the data power of each user should satisfy the individual power budget, hence we have
\vspace*{-0.2cm}
\begin{equation} \label{eq:PowerConstraint}
 I_k (\pmb{\rho}) \leq \rho_k \leq P_{\max,k}.
\vspace*{-0.2cm}
\end{equation}
One can search across the range of each data power variable observed in  \eqref{eq:PowerConstraint}, where the global optimum  of Problem~\eqref{Problem:Epigraphform} is validated by Theorem~\ref{Theorem:Bisection}.
\vspace*{-0.2cm}
\begin{theorem} \label{Theorem:Bisection}
For a given feasible $\xi_o$ value and the initial data powers $\rho_k(0) = P_{\max,k}, \forall k$, the globally optimal solution of Problem~\eqref{Problem:TotalTransmitPower} is obtained by computing the standard interference function in \eqref{eq:Ikrho} and the power constraint in \eqref{eq:PowerConstraint} for all users. In more detail, if the data power of user~$k$ is updated at iteration~$n$ as   
\vspace*{-0.2cm}
\begin{equation} \label{eq:rho}
\rho_k(n) = I_k(\pmb{\rho}(n-1)),
\vspace*{-0.2cm}
\end{equation}
where $I_k(\pmb{\rho}(n-1))$ is defined in \eqref{eq:Ikrho} with $\pmb{\rho}(n-1)$ denoting the data power vector from the previous iteration, then this iterative approach converges to the unique optimal solution after a finite number of iterations.  Owning to the feasibility of $\xi_o$, it holds that $I_k(\pmb{\rho}(n-1)) \leq P_{\max,k}, \forall k$.

Then, the most energy-efficient solution of Problem~\eqref{Problem:Epigraphform} is obtained by updating the lower bound of the SINR values across the search range $\xi_o \in [0,  \xi_o^{\mathrm{up}}]$, where $\xi_o^{\mathrm{up}}$ is given by
\vspace*{-0.2cm}
\begin{equation} \label{eq:xioup}
\xi_o^{\mathrm{up}} =  \underset{k}{\min} \, \frac{P_{\max,k} \big|\|\bar{\mathbf{g}}_k\|^2 +  p K \mathrm{tr}(\pmb{\Theta}_k)  +  \sum_{m=1}^M \gamma_{mk} \big|^2}{ \mathsf{NO}_k}.
\vspace*{-0.1cm}
\end{equation}
Along all considered values of the variable $\xi_0$, Problem~\eqref{Problem:TotalTransmitPower} is infeasible if the following condition is met at least by one user for a given value $\xi_o$ as
\vspace*{-0.2cm}
 \begin{equation} \label{eq:IkPmaxk}
 R_k (\pmb{\rho}(n)) < B(1-K/\tau_c)\log_2(1+\xi_o), \exists k \in \{1, \ldots, K\}.
 \vspace*{-0.2cm}
 \end{equation}
\end{theorem}
\begin{proof}
The main proof hinges on verifying the standard interference function defined for each user and on finding an upper bound for the bisection method. The detailed proof is omitted due to space limitations.
\end{proof}

From the analytical features in Theorem~\ref{Theorem:Bisection}, the proposed alternating technique of finding the optimal solution to Problem~\eqref{Problem:MaxMinQoS} is shown in Algorithm~\ref{Algorithm1} by initially setting the maximum data power to each user, i.e., $\rho_k(0) = P_{\max,k}, \forall k$ and the range of $[\xi_{\min,o}, \xi_{\max,o}]$ for the parameter $\xi_o$ by capitalizing on \eqref{eq:xioup}. The bisection is utilized to update $\xi_o$, while the data powers are iteratively updated by the standard interference function seen in \eqref{eq:rho} subject to the condition \eqref{eq:PowerConstraint}. In particular, for a given value of $\xi_o = (\xi_{\min,o} + \xi_{\max,o})/2$, the temporary data power coefficients  are set as $\tilde{\rho}_k(0) = \rho_k(0), \forall k$. Then user~$k$ will update its temporary data power at inner iteration~$m$ as
\vspace*{-0.1cm}
\begin{equation} \label{eq:rhotilde}
	\tilde{\rho}_k(m) = \min({I}_{k} \left(\tilde{\pmb{\rho}} (m-1) \right), P_{\max,k}).
\vspace*{-0.1cm}
\end{equation}
The inner loop can be terminated, when the difference between two consecutive iterations is small. For example, we may compute the normalized total power consumption ratio
\vspace*{-0.1cm}
\begin{equation} \label{eq:gamman}
T  = \frac{\left|\sum_{k=1}^K \tilde{\rho}_{k}(m)  - \sum_{k=1}^K \tilde{\rho}_{k}(m-1)  \right|}{\sum_{k=1}^K \tilde{\rho}_{k}(m-1) }.
\vspace*{-0.1cm}
\end{equation}
The data power will converge to the optimal solution with tolerance as $\gamma(n) \leq \epsilon$. For a given value $\xi_o$, if Problem~\eqref{Problem:TotalTransmitPower} is feasible, the lower bound of $\xi_o$ is then updated, yielding $\xi_{\min,o} = \xi_o$ after obtaining the data power solution. Otherwise, \eqref{eq:IkPmaxk} is utilized  to detect if there is no solution to Problem~\eqref{Problem:TotalTransmitPower} for a given $\xi_o$. The closed-form expression in updating the temporary data power as shown in \eqref{eq:rhotilde} will be the data power of user~$k$ if $R_k(\tilde{\pmb{\rho}}_{k}(m)) \geq R_o, \forall k$. It means that outer iteration~$n$ will perform  $\rho_k(n) = \tilde{\rho}_k(m)$ 
and update $\xi_{\min,o} = \xi_o$. Otherwise, the upper bound $\xi_{\max,o}$ will be shrunk as $\xi_{\max,o} =\xi_o$. Assuming that the dominant arithmetic operators are multiplications and divisions, we can estimate the computational complexity of Algorithm~\ref{Algorithm1} is of the order of $
C_1 = \mathcal{O}\left( \lceil \log_2 (\xi_o^{\mathrm{up}}/\delta) \rceil  \left( (4K^2 + MK^2 + 5K)(U_1 + 1)\right) \right)$
where $ \lceil \cdot \rceil $ is the ceiling function. The computational complexity is directly proportional to the number of APs and it is in a quadratic function of the number of users. However,  the computational complexity of our proposed algorithm does not depend on the number of satellite antennas.
\begin{algorithm}[h]
	\caption{Data power allocation to Problem~\eqref{Problem:MaxMinQoS}} \label{Algorithm1}
	\textbf{Input}:  Define  $P_{\max,k}, \forall k$; Select  $\rho_{k}(0) = P_{\max,k}, \forall k$; Set $\xi_{o}^{\mathrm{up}}$ as in \eqref{eq:xioup}; Define $\xi_{\min,o} = 0$ and $\xi_{\max,o} = \xi_{o}^{\mathrm{up}}$; Set the inner tolerance $\epsilon$ and the outer tolerance $\delta$.
	\begin{itemize}[leftmargin=*]
		\item[1.]  Initialize the outer loop index $n=1$.
		\item[2.] \textbf{while} $\xi_{\max,o} - \xi_{\min,o} > \delta$ \textbf{do}
		\begin{itemize}
			\item[1.1.] Set $\tilde{\rho}_k(0) = \rho_k(0), \forall k$; Set $\xi_o = (\xi_{\min,o}  + \xi_{\max,o} )/2$ and compute $R_o =  B(1-K/\tau_c)\log_2(1+\xi_o)$.
			\item[1.2.] Compute $P_{\mathrm{tot}}(0) =\sum_{k=1}^K \rho_{k}(0)$.
			\item[1.3] Initialize the accuracy $T= P_{\mathrm{tot}}(0)$ and set $m=1$.
			\item[1.4.] \textbf{while} $T> \varepsilon$ \textbf{do}
			\begin{itemize}
				\item[1.4.1.] User~$k$ computes ${I}_{k} \left(\tilde{\pmb{\rho}} (m-1) \right)$ using \eqref{eq:rho}.
				\item[1.4.2.] User~$k$ updates its temporary data power as \eqref{eq:rhotilde}.
				\item[1.4.3.] Repeat Steps $1.4.1$ and $1.4.2$ with other users, then update the accuracy as in \eqref{eq:gamman}.
				\item[1.4.4.] If $T \leq \varepsilon$, compute $R_k(\tilde{\pmb{\rho}}_{k}(m)),\forall  k,$ and go to Step~$1.6$. Otherwise, set $m= m+1$ and go to Step $1.4.1$.
			\end{itemize}
			\item[1.5.] \textbf{End while}
			\item[1.6.] If $\exists k, R_k(\tilde{\pmb{\rho}}_{k}(m)) < R_o$, set $\xi_{\max,o} = \xi_o$ and go to Step 1. Otherwise, update $\rho_{k}(n) = 	\tilde{\rho}_k(m), \forall k,$ and set $\xi_{\min,o} = \xi_o$, set $n=n+1$, and go to Step~$1$.
			
		\end{itemize}
		\item[3.] \textbf{End while}
		\item[4.]  Set $\rho_{k}^{\ast} = \rho_{k}(n),\forall  k$.
	\end{itemize}
	\textbf{Output}: Final interval $[\xi_{\min,o}, \xi_{\max,o}]$ and $\{ \rho_{k}^{\ast} \}$, $\forall l,k$. \vspace*{-0.0cm}
\end{algorithm}
\begin{figure*}[t]
	\begin{minipage}{0.24\textwidth}
		\centering
		\includegraphics[trim=0.8cm 0.0cm 1.4cm 0.6cm, clip=true, width=1.74in]{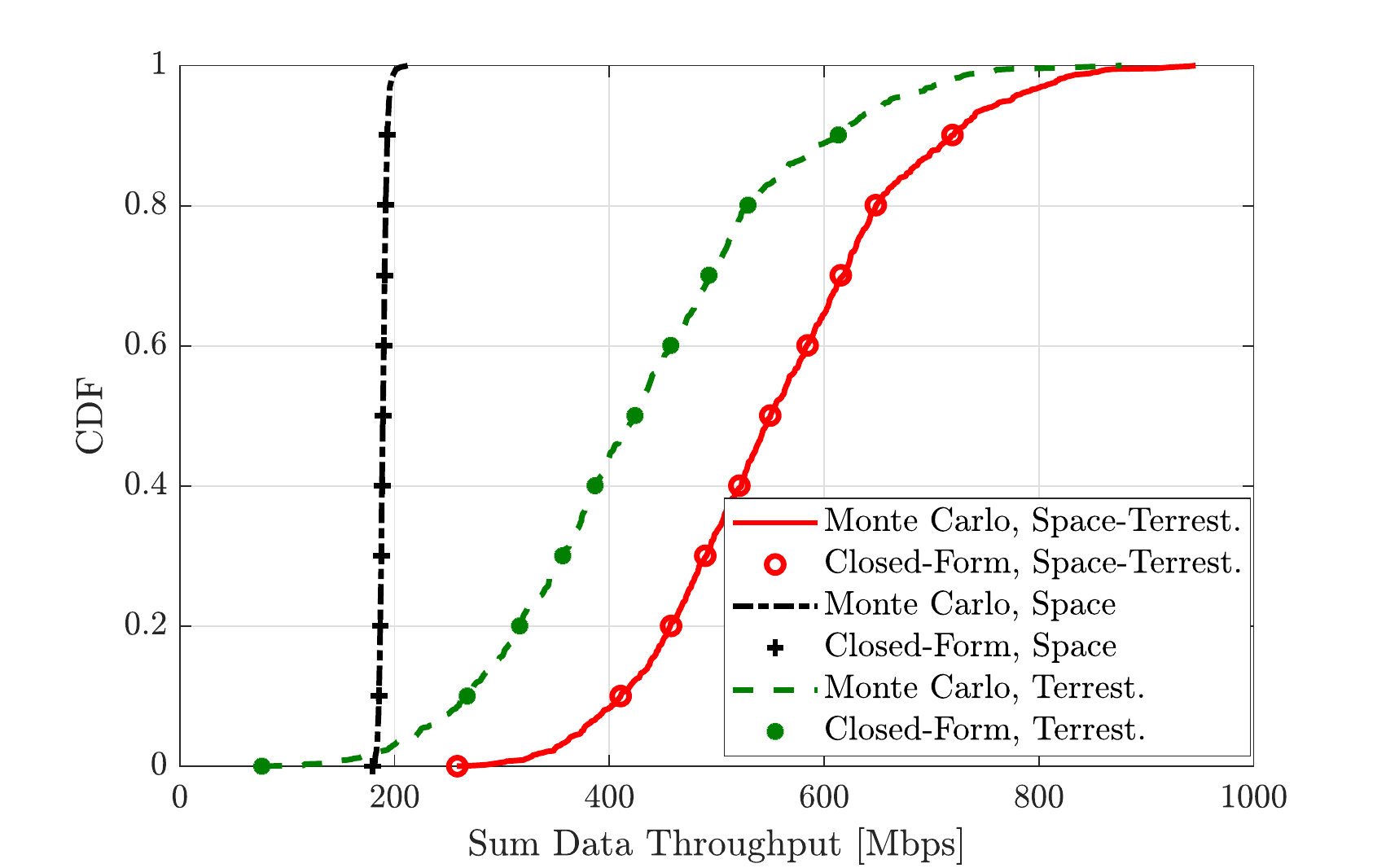} \vspace*{-0.5cm}\\
		\fontsize{8}{8}{$(a)$}
		\vspace*{-0.3cm}
	\end{minipage}
	\hfil
	\begin{minipage}{0.24\textwidth}
		\centering
		\includegraphics[trim=0.8cm 0.0cm 1.4cm 0.6cm, clip=true, width=1.78in]{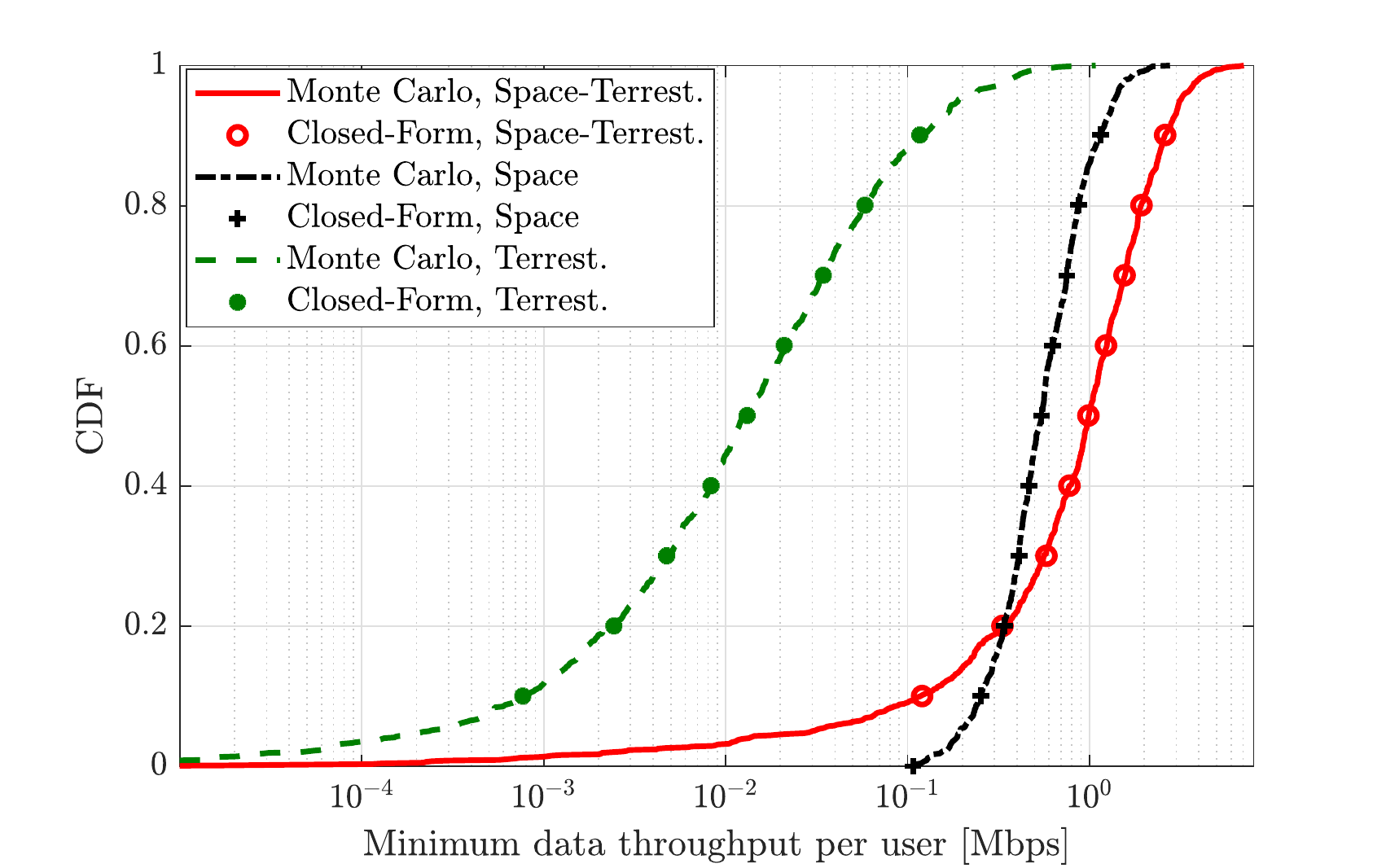} \vspace*{-0.5cm} \\
		\fontsize{8}{8}{$(b)$}
		\vspace*{-0.3cm}
	\end{minipage}
\hfil
\begin{minipage}{0.24\textwidth}
	\centering
	\includegraphics[trim=0.8cm 0.0cm 1.4cm 0.6cm, clip=true, width=1.8in]{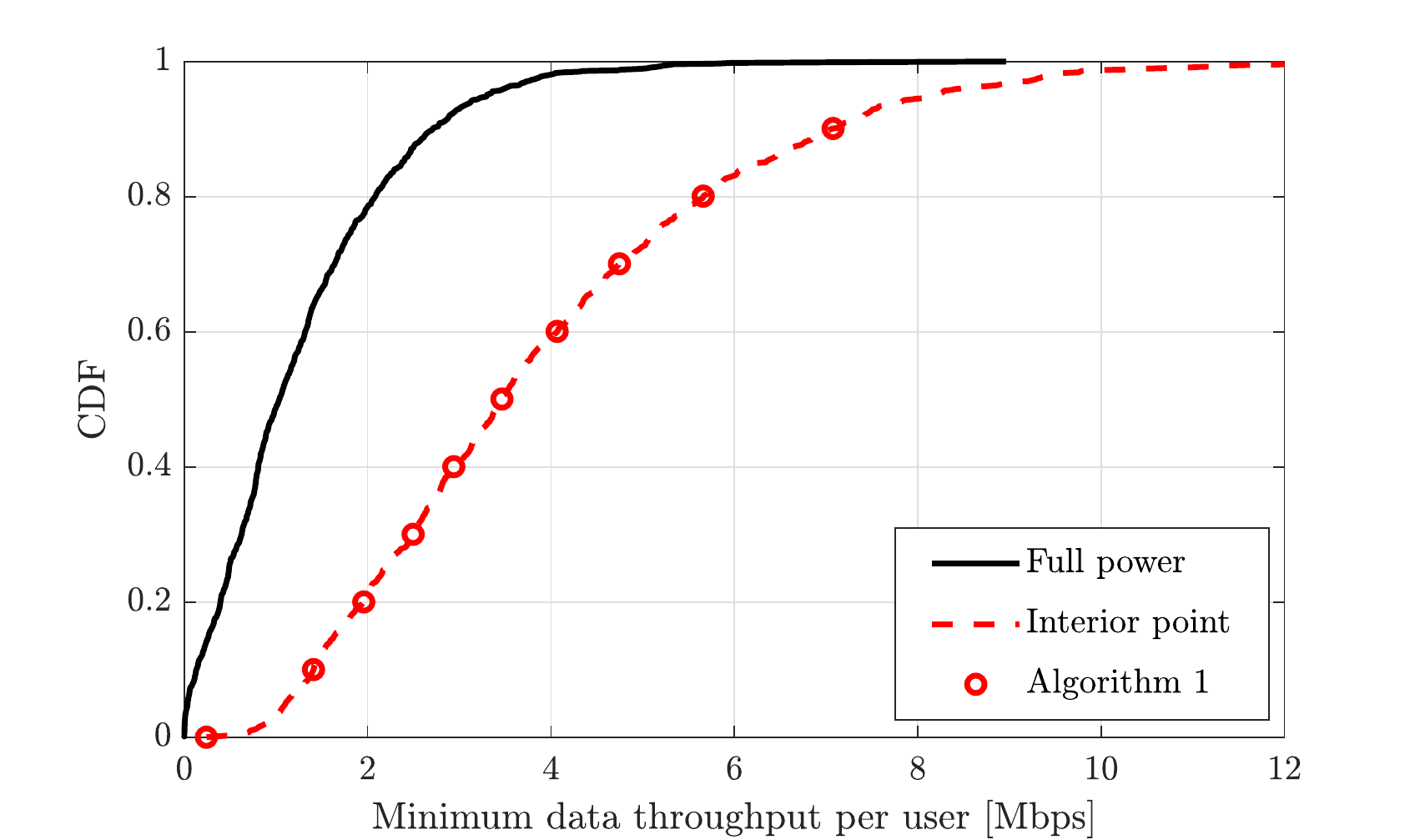} \vspace*{-0.5cm} \\
	\fontsize{8}{8}{$(c)$}
	\vspace*{-0.3cm}
\end{minipage}
\hfil
\begin{minipage}{0.24\textwidth}
	\centering
	\includegraphics[trim=0.8cm 0cm 1.4cm 0.6cm, clip=true, width=1.74in]{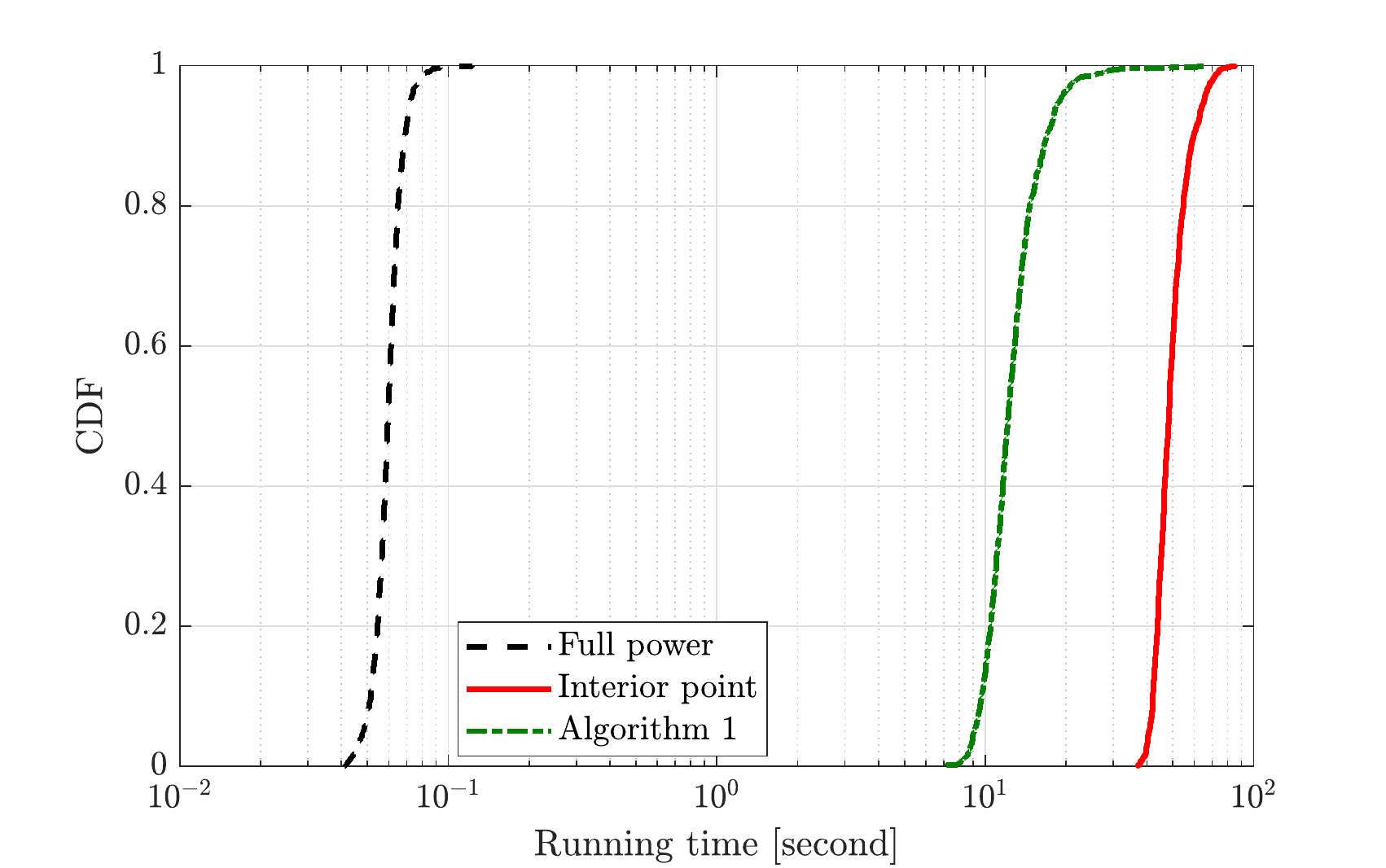}  \vspace*{-0.5cm} \\
	\fontsize{8}{8}{$(d)$}
	\vspace*{-0.2cm}
\end{minipage}
\caption{The system performance: $(a)$ CDF of the sum ergodic data throughput [Mbps] using Monte Carlo simulations and the analytical frameworks; $(b)$ CDF of the minimum data throughput per user [Mbps] using Monte Carlo simulations vs the analytical frameworks; $(c)$ CDF of the minimum data throughput per user [Mbps] for the space-terrestrial communication system; $(d)$ CDF of the running time to obtain the max-min fairness solution for the space-terrestrial communication system.}  \label{Fig2}
\vspace*{-0.6cm}
\end{figure*}
\vspace*{-0.2cm}
\section{Numerical Results} \label{Sec:NumericalResults}
\vspace*{-0.2cm}
We consider a network's deployment in a rural area with $40$ BSs distributed in a square area of $20$~km$^2$, mapped into a Cartesian coordinate system $(x,y,z)$. A LEO satellite is equipped with $N=100$, ($N_H = N_V = 10$), antennas and it is located at the position $(300, 300, 400)$~km. 
The antenna gain at the terrestrial devices is $5.0$~dBi and it is $26.9$~dBi at the satellite. The system bandwidth is $B=100$~MHz and the carrier frequency is $f_c = 20$~GHz. The coherence block is $\tau_c = 5000$ OFDM subcarriers. The transmit power assigned to each data symbol is $5$~dBW \cite{abdelsadek2021future}. The noise figure at the BSs and satellite are $7$~dB and $1.2$~dB, respectively.  The large-scale fading coefficient between user~$k$ and BS~$m$ is  suggested by the 3GPP model \cite{LTE2017a}, as $\beta_{mk} = G_m + G_k  - 8.50 -  20\log_{10}(f_c) -  38.63 \log_{10} (d_{mk})  + \zeta_{mk}, $
where $G_m$ and $G_k$ are the antenna gains at AP~$m$ and user~$k$, respectively. The distance between this user and AP~$m$ is denoted as $d_{mk}$ and $\zeta_{mk}$ is the shadow fading following a log-normal distribution with standard derivation $8$~dB. The large-scale fading coefficient between user~$k$ and the satellite is defined  \cite{3gpp2019study} as $
	\beta_{k} = G + G_k + \tilde{G}_k -32.45 - 20 \log_{10} (f_c) - 20 \log_{10} (d_k) + \zeta_k, $
where $G$ is the RA gain at the satellite and its normalized beam pattern is $\tilde{G}_k = 4 \left|J_1 \left( \frac{2\pi}{\lambda} \alpha \sin(\phi_k) \right)/\left(\frac{2\pi}{\lambda} \alpha \sin(\phi_k) \right) \right|^2$ if $ 0 \leq \phi_k \leq \pi/2$. Otherwise, $\tilde{G}_k =0$ if $\phi_k = 0$. Here, $\alpha$ denotes the radius of the antenna's circular aperture; $\lambda$ is the wavelength; and $\phi_k$ is the angle between user~$k$ and its beam center. The shadow fading $\zeta_k$ is obtained from a log-normal distribution with the standard derivation depending on the carrier frequency, channel condition, and the elevation angle \cite{3gpp2019study}. $d_k$~[m] is the distance between the satellite and user~$k$.
We consider $1000$ different time slots, each consisting of $20$ users uniformly located in the coverage area.  The three different systems are considered: $i)$ \textit{the space-terrestrial system} represented by the SINR expression in \eqref{eq:SINRk} with the overall channel coefficient $z_{kk'}$ in \eqref{eq:zkkprime} and the analytical framework in \eqref{eq:ClosedSINR}; $ii)$ \textit{the stand-alone terrestrial system} represented by the SINR expression in \eqref{eq:SINRk} with $z_{kk'}=   \sum\nolimits_{m=1}^M  u_{mk}^\ast  g_{mk'} $; and $iii)$ \textit{the stand-alone space  system} represented by the SINR expression in \eqref{eq:SINRk} with $z_{kk'}=  \mathbf{u}_k^H \mathbf{g}_{k'} $.

Fig.~\ref{Fig2}$(a)$ compares the cumulative distribution function (CDF) of the sum data throughput. The numerical simulations and the analytical results match very well for the different systems that validates our analysis. A network only relying on the satellite offers quite stable throughput on average. The terrestrial system offers $2.3$~times better the sum throughput than the baseline. Jointly processing the received signals, the space-terrestrial system supports superior improvements of $30\%$ in throughput over utilizing the APs only. Fig.~\ref{Fig2}$(b)$ shows the CDF of the minimum throughput where the terrestrial system is the baseline. The satellite system yields $14$ times higher than the baseline. Superior gains up to  $28.8$ times better than only using the APs are obtained by integrating the satellite into a terrestrial network.

Fig.~\ref{Fig2}$(c)$ plots the CDF of the throughput of the different power allocation strategies. The full data power transmission produces the lowest average max-min fairness level, which is $1.3$~[Mbps]. The two remaining algorithms generate the same solution that is $3$ times better than the full power transmission on average.  Indeed, the interior point methods have been widely applied for solving the max-min fairness optimization problem \cite{ngo2017cell}. The associated running time required to obtain the max-min fairness solution is given in Fig.~\ref{Fig2}$(d)$. The running time is only $0.06$~[s] dedicated to estimating the data throughput if the system allows each user transmit at full power per data symbol. The interior-point methods spend $50$~[s] to obtain the solution of  Problem~\eqref{Problem:MaxMinQoS}, while Algorithm~\ref{Algorithm1} provides a rich time reduction factor of $3.8$   times.

\vspace*{-0.2cm}
\section{Conclusions} \label{Sec:Concl}
\vspace*{-0.1cm}
The throughput analysis and the max-min data power control of a multi-user system were provided in the presence of an NGSO satellite and distributed APs. We assumed a centralized signal processing unit for boosting the throughput per user for coherent data detection combining the received signals of the space and terrestrial links. The achievable data throughput expression derived can be applied to an arbitrary channel model and combining techniques. A closed-form expression was also derived for the MRC receiver technique and spatially correlated  channels. 
The satellite boosts the sum throughput in the network by more than $30\%$ for the  parameter setting considered, while the minimum  throughput is enhanced by more than tenfold. 
\bibliographystyle{IEEEtran}
\bibliography{IEEEabrv,refs}
\end{document}